%% file: arxiv.tex
\documentclass{article}

\usepackage[utf8]{inputenc} 
\usepackage[T1]{fontenc}    
\usepackage{url}            
\usepackage{booktabs}       
\usepackage{amsfonts}       
\usepackage{nicefrac}       
\usepackage{microtype}      
\usepackage{lipsum}
\usepackage{graphicx}

\usepackage{algorithm}
\usepackage{algorithmic}
\usepackage[switch]{lineno}
\usepackage{dirtytalk}
\usepackage{mathtools}
\usepackage{amsthm}
\usepackage{authblk}

\usepackage{natbib}
\usepackage[verbose=true,letterpaper]{geometry}
\AtBeginDocument{
  \newgeometry{
    textheight=9in,
    textwidth=6.5in,
    top=1in,
    headheight=14pt,
    headsep=25pt,
    footskip=30pt
  }
}

\usepackage[dvipsnames]{xcolor}
\usepackage[colorlinks, citecolor=teal, linkcolor=RedViolet, urlcolor=teal]{hyperref}
\usepackage[nameinlink]{cleveref}

\setcitestyle{authoryear, open={(},close={)}}
 
\newtheorem{theorem}{Theorem}
\newtheorem{lemma}{Lemma}
\newtheorem{definition}{Definition}
\newtheorem{observation}{Observation}
\newtheorem{corollary}{Corollary}
\newtheorem{claim}{Claim}

\title{On The Pursuit of EFX for Chores:\\ Non-Existence and Approximations}
\author[1,2]{
Vasilis Christoforidis
}
\author[1,3]{
Christodoulos Santorinaios
}
\affil[1]{Archimedes/Athena RC}
\affil[2]{Aristotle University of Thessaloniki}
\affil[3]{Athens University of Economics and Business}

\begin{document}
    \maketitle
    \begin{abstract}
        \input{abstract}
    \end{abstract}
    
    \input{intro}
    \input{contributions}
    \input{related}
    \input{prelims}
    \input{mainRes}
    \input{fewItems}
    \input{addApx}
    \input{conclusion}
    \input{acknowledgments}

    \bibliography{references}
  
\end{document}

%% file: abstract.tex
We study the problem of fairly allocating a set of chores to a group of agents. The existence of envy-free up to any item (EFX) allocations is a long-standing open question for both goods and chores. We resolve this question by providing a negative answer for the latter, presenting a simple construction that admits no EFX solutions for allocating six items to three agents equipped with superadditive cost functions, thus proving a separation result between goods and bads. In fact, we uncover a deeper insight, showing that the instance has unbounded approximation ratio. Moreover, we show that deciding whether an EFX allocation exists is NP-complete. On the positive side, we establish the existence of EFX allocations under general monotone cost functions when the number of items is at most $n + 2$.  We then shift our attention to additive cost functions. We employ a general framework in order to improve the approximation guarantees in the well-studied case of three additive agents, and provide several conditional approximation bounds that leverage ordinal information.

%% file: intro.tex
\section{Introduction}\label{sec:intro}

Fair Division has been widely studied in the past decade, yielding a series of results for various fairness notions. One of the most popular notions is \emph{envy-freeness} (EF), under which each agent (weakly) prefers her own bundle to any other agent's bundle. In the case of divisible items, an EF allocation is always guaranteed to exist, while for indivisible items this is not always the case; consider for instance the scenario where we have to allocate a single item among two agents. This fact has led to numerous relaxations of envy-freeness and approximations thereof. 

One such notion is that of \emph{envy-freeness up to one item} (EF1) \cite{Budish2011}. In EF1 allocations, an agent $i$ might envy agent $j$ but the envy is eliminated after hypothetically removing some item; either a good from agent $j$'s bundle or a chore from agent $i$'s bundle. EF1 allocations are known to exist for both goods \cite{LiptonMarkakisMosselSaberi} and chores \cite{BhaskarSVapprox,AzizCaragiannis}.

A stronger variant is that of \emph{envy-freeness up to any item} (EFX); an allocation is said to be EFX if envy vanishes after the removal of \emph{any} item. The existence of EFX allocations remains a challenging open problem in the area and has been even deemed as \say{fair division's most enigmatic question} \cite{ProcacciaEnigmatic}. EFX is known to exist for special cases: two agents with general valuations \cite{PlautRoughgarden}, and three additive agents \cite{ChaudhuryGargMehlhorn}\footnote{More recently, the result was further extended to capture more general valuations via a simplified analysis \cite{BergerCohenFeldmanFiat,AkramiEFXsimplified}.}. In contrast to the fruitful agenda on EFX for goods, the landscape of EFX allocations is less explored in the context of chores. For instance, even the existence for three additive agents as well as a constant factor approximation is elusive.

The problem of EFX allocations remains open, even for additive valuations, in the context of goods, \say{despite significant effort} \cite{CaragiannisUnreasonable}. Moreover, it has been suggested by \citeauthor{PlautRoughgarden} that EFX allocations may fail to exist:

\begin{quote}
    \emph{We suspect that at least for general valuations, there exist instances where no EFX allocation exists, and it may be easier to find a counterexample in that setting.}
\end{quote}

We verify the aforementioned suspicion, answering the analogous question of whether an EFX allocation always exists for chores to the negative. No such counterexample was known even in the setting of general monotone valuations, either for the goods or the chores setting. In fact, the only known counterexamples for the non-existence of EFX allocations employ \emph{mixed manna}, i.e., mixtures of goods and chores, therefore non-monotone valuations \cite{bérczi2020envyfree,HosseiniMixturesLex}. Apart from being the first counterexample for EFX over general monotone functions, our chore construction signifies the first separation from its goods-only counterpart.

%% file: contributions.tex
\subsection{Our Contributions}
We study fair allocations in a setting where $m$ indivisible chores need to be allocated to $n$ agents in a fair manner. We focus on a well-studied notion of fairness, (approximate) envy-freeness up to any item. 

We consider the following as our main technical results:

\begin{itemize}
    \item An EFX allocation for chores need not exist under general cost functions. We present a construction with three agents in which no bounded approximation exists (\Cref{th:non-ex}). 
    \item Determining whether an instance with three agents and superadditive costs admits an EFX allocation is NP-complete (\Cref{thm:npc}).
\end{itemize}

 We note that no such counterexample was previously known for any subset of general monotone valuations, either in the context of goods-only or chores-only. Notably, this is the first separation result between goods and chores regarding EFX, since EFX allocations are known to exist when the number of goods is at most $n+3$ \cite{Mahara21esa}. Lastly, we extend our non-existence and hardness results to many agents (\Cref{th:nagents}).\\

In light of these negative results, we focus on a setting with few items, namely $m \leq n + 2$; we prove the existence of EFX allocations under general monotone cost functions (\Cref{th:n+2}). Due to the aforementioned negative example, this is the largest constant $c$ for which all instances with three agents and $m \leq n + c$ items admit an EFX allocation. This is the first nontrivial result for a small number of chores under general cost functions; similar results have been established for goods \cite{AmanatidisMarkakisNtokos,Mahara21esa}, as well as for chores, albeit under additive utilities \cite{KobayashiMS23}.

Next, we focus on additive cost functions and adapt a general framework in order to obtain approximation guarantees for chores (\Cref{theorem:frame}), establishing a series of improved (conditional) approximation ratios under ordinal-based assumptions. We follow recent works due to \cite{BhaskarSVapprox,LiLiWuChores} that employ a variant of the well-known \emph{Envy Cycle Elimination} technique, namely \emph{Top Trading Envy Cycle Elimination}, to obtain improved approximation guarantees. Finally, we switch to the special case of three agents equipped with additive valuations. We improve the approximation ratio from $2+\sqrt{6}$ to 2 (\Cref{th:3agents}).

%% file: related.tex
\subsection{Related Work}
In this section we discuss prior works regarding EFX for goods and chores. We focus on the latter case. The growing literature on fair division is too extensive to cover here, and thus, we point the interested reader to the survey of \citeauthor{AmanatidisSurvey} for an extensive discussion on recent developments, along with further notable fairness notions and open problems. 

The seminal work of \cite{CaragiannisUnreasonable} showed that maximizing the Nash social welfare produces EF1 and Pareto optimal allocations for goods. The existence of EF1 and PO allocations remains a major open problem for chores, beyond a couple of restricted settings
\cite{Garg_Murhekar_Qin_2022,EbadianPSChores,BarmanNVSupermodular}.

\paragraph{Envy-freeness up to any item (EFX) for goods} Perhaps the most compelling relaxation of envy-freeness is EFX \cite{GourvesMonnotTlilane,CaragiannisUnreasonable}. In sharp contrast to EF1 that enjoys strong existential and algorithmic properties, EFX remains a challenging open problem. In the past years, numerous works have studied approximate versions while also establishing the existence of the notion in restricted settings. \cite{PlautRoughgarden} considered approximate EFX showing that $1/2$-EFX allocations always exist; \cite{ChanMMA} subsequently showed that such allocations can be computed in polynomial time  while \citeauthor{AmanatidisMarkakisNtokos} improved the approximation ratio to $\phi - 1$ for additive valuations, which is the best currently known factor. \citeauthor{ChaudhuryGargMehlhorn} showed in a breakthrough result that exact EFX allocations always exist for three agents with additive valuation functions. Regarding restricted settings, positive results are known for a small number of items, lexicographic preferences, two types of goods, two valuation types, and EFX in graphs \cite{Mahara21esa,HosseiniSVX21LexGoods,MaharaTypes,gorantla2023fair,ChristodoulouFiatKoutsoupiasSgouritsa}. Lastly, a major line of work has focused on binary valuations and generalizations thereof, including bi-valued instances and dichotomous valuations \cite{HalpernBinaryRule,AmanatidisEFXstories,Babaioff_Ezra_Feige_2021,BenabbouChakrabortyIgarashiZick}.
\paragraph{Envy-freeness up to any item (EFX) for chores}
In contrast to the case of goods, the existence of EFX allocations even for three agents with additive valuations remains an open problem. \cite{ZhouWuChores} obtained a $5$-approximation (later improved to $2+\sqrt{6}$ in the journal version) while also showing that an $O(n^2)$-EFX allocation always exists under additive cost functions for any number of agents. \cite{LiLiWuChores} showed the existence of EFX allocations when agents exhibit identical orderings over the set of items (commonly referred to as IDO instances), while \cite{GafniCopies} showed the existence of EFX allocations under additive leveled valuations. Similarly to the case of goods, several works have shown positive results for dichotomous valuations \cite{ZhouWuChores,KobayashiMS23,BarmanNVSupermodular,tao2023efx_poChores}. EFX allocations always exist under additive cost functions when $m \leq 2n$ \cite{KobayashiMS23} or when there are only two types of chores \cite{AzizTwoTypesChores}.

Lastly, we note that we heavily rely on an important subclass of valuations, namely superadditive cost functions; such functions capture complementarities among items and have received significant attention in the microeconomics and game theory literature \cite{AGTBook,HassidimKMN}. Prior work has also examined fair allocations in the presence of complements, both in the goods and the bads setting \cite{CaragiannisUnreasonable,BarmanNVSupermodular}.

\subsection{Paper Outline}

The remaining sections of the paper are outlined below. \Cref{sec:prelim} includes the formal model and relevant definitions. In \Cref{sec:non-existence} we derive our main technical results regarding existence, approximation, and hardness of EFX allocations. \Cref{sec:few-items} deals with the few items setting, while in \Cref{sec:addapx} we show improved approximations under additive cost functions. Finally, we conclude and propose two major open questions.

%% file: prelims.tex
\section{Preliminaries}\label{sec:prelim}

The problem of discrete fair division with chores is described by the tuple $\langle N, M, C \rangle$ where $N = \{1,\dots, n\}$ is the set of $n$ agents, $M$ is the set of $m$ indivisible chores and $C = \left(c_1(\cdot), \dots, c_n(\cdot) \right)$ is the agents' cost functions. 

\paragraph{Cost functions.} For each agent $i$, $c_i: 2^M \to \mathbb{R}_{\ge 0}$ is normalized, i.e. $c_i(\emptyset) = 0$, and monotone, $c_i(S \cup \{e\}) \ge c_i(S)$ for all $S \subseteq M$ and $e \in M$. A cost function $c$ is superadditive if for any $S, T \subseteq M: c(S \cup T) \geq c(S)+c(T)$ and additive if the previous relation holds always with equality. For ease of notation, we sometimes use $e$ instead of $\{e\}$. We use the terms valuations and cost functions for chores (or bads) interchangeably.

\paragraph{Allocation and bundles.} A subset of chores $X \subseteq M$ is called a bundle. An allocation $X = (X_1, \dots, X_n)$ is an $n$-partition of $M$, i.e., $X_i \cap X_j = \emptyset, \forall i,j \in N$ and $\bigcup_{i \in N} X_i = M$, in which agent $i$ receives the bundle $X_i$. We denote by $\sigma_i(j, S)$ the $j$-th most costly chore in $S$ under $c_i$, with ties broken arbitrarily. For instance, we may write $c_i(\sigma_i(1), M) \ge c_i(\sigma_i(2), M) \ge \dots \ge c_i(\sigma_i(m), M)$ to describe agent $i$'s preference over all chores in $M$.  When it is clear from context we will drop the set parameter for brevity.

\paragraph{Envy-freeness.} An allocation is said to be envy-free if there is no envy among agents. It is envy-free up to one item (EF1) if for any pair of agents $i, j \in N$ where $i$ envies $j$, there exists some item $e \in X_i$ such that $c_i(X_i \setminus e) \le c_i(X_j)$. An allocation $X = (X_1, \dots, X_n)$ is called EFX if for all pairs of agents $i,j \in N$ and any chore $e \in X_i$ it holds that $c_i(X_i \setminus e) \le c_i(X_j)$. Lastly, for approximate EFX allocations we have the following definition. 

\begin{definition}[$\alpha$-EFX]\label{def:efx}
    An allocation $X$ is $\alpha$-approximate envy free up to any chore ($\alpha$-EFX) if for any pair of agents $i,j$ and any $e \in X_i:\ c_i(X_i \setminus e) \le \alpha \cdot c_i(X_j)$. 
\end{definition}
We say that an agent $i$ strongly envies when \Cref{def:efx} is violated, i.e., $c_i(X_i \setminus e) > \alpha \cdot c_i(X_j)$ for some $e \in X_i$. By setting $a=1$ we retrieve the definition of an exact EFX allocation\footnote{We note that there exists an alternative definition in the literature, in which $\alpha \cdot c_i(X_i \setminus e) \le c_i(X_j)$ for any pair of agents $i,j$ and any $e \in X_i$. In this case, $\alpha$ lies within the same range as in the context of EFX approximations for goods, i.e., $0 < \alpha \leq 1.$}. 

\paragraph{Maximin share.} An extensively studied notion in fair division is maximin share fairness, introduced by \cite{Budish2011}. In the context of chores, the maximin share (MMS) corresponds to the minimum value an agent could guarantee to herself after partitioning the items into $n$ bundles and keeping the most burdensome thereof. MMS allocations need not exist even under additive cost functions, but strong approximation guarantees are known \cite{HuangLuMMS}.

\begin{definition} [Maximin share]
An allocation $X$ is said to be maximin share fair (MMS) if $$c_i(X_i) \leq \mu_i^n(M) = \min_{X \in \Pi_n(M)} \max_{k \in [n]} c_i(X_k), \quad \forall i \in N$$
\end{definition}

\subsection{Top Trading Envy Cycle Elimination}
In \Cref{sec:addapx}, en route to obtaining better approximation guarantees we will make use of the \textit{Top Trading Envy Cycle Elimination} algorithm (TTECE,  \Cref{alg:ttece}). Therefore, we think it is useful to include a short description. In contrast to the goods-only setting where the utilities of the agents are non-decreasing while performing envy-cycle eliminations, their cost increases when picking an item while decreases for the agents involved in a cycle elimination. The main insight of the algorithm is that an agent that does not envy any other agent in some allocation $X$, meaning that $X_i$ is her \textit{top} bundle, can receive an additional chore without violating the EF1 property. If such an agent always exists, then we can proceed in an incremental fashion, allocating one item at a time. Assuming that no such agent exists we can create the envy digraph of the allocation $G_X$ as follows: each node represents an agent and an edge from node  $i$ to node $j$ represents that agent $j$ owns $i$'s top bundle. Since every node has an outgoing edge the graph contains a cycle $C$. Reallocating the bundles along that cycle, i.e. each agent receives the bundle she envies, creates a new allocation $X_C$ that maintains the EF1 property and creates unenvied agents (sinks).

\begin{algorithm}[tb]
    \caption{Top trading envy cycle elimination algorithm}
    \label{alg:ttece}
    \textbf{Input}: $N, M, C$\\
    \textbf{Output}: An allocation $X$
    \begin{algorithmic}[1]
        \STATE $X=\{\emptyset,\dots, \emptyset\}$.
        \WHILE{$\exists$ some unallocated chore $e$}
        \IF {there is no sink in $G_X$}
        \STATE Find a cycle $C$ in $G_X$
        \STATE $X = X^C$
        \ENDIF
        \STATE Choose a sink agent $s$ and $X_s = X_s \cup e$
        \ENDWHILE
        \STATE \textbf{return} $X$
    \end{algorithmic}
\end{algorithm}

%% file: mainRes.tex
\section{Non-Existence, Hardness, and Inapproximability of EFX}\label{sec:non-existence}

In this section we present our main technical results; namely, we describe our explicit construction and proceed with showing that no finite EFX approximation is possible. Then, we show that deciding whether an EFX allocations always exists is NP-complete.

\subsection{Non-existence and Inapproximability}\label{subsec:31}
Our negative example relies on a simple superadditive structure with three \say{special} chores, which are common for all agents. We use repeatedly the fact that an agent $i$ valuing the bundle of another agent $j$ at zero, i.e. $c_i(X_j)=0$, can afford to take at most one item or a bundle of zero value, i.e. $c_i(X_i)=0$; this follows from the definition of EFX in the context of chores.

\begin{theorem}\label{th:non-ex}
    An EFX allocation need not exist for three agents with superadditive cost functions. Moreover, no approximate solution exists, for any approximation factor.
\end{theorem}

\begin{proof}
The set of chores consists of $\{\hat a, a_1, a_2, b_1, b_2, b_3\}$ and the agents have identical costs for single chores, as given in the table below\footnote{We will at times abuse notation, using $a$ as the name of $a_1$ and $a_2$ at once for ease of exposition.}:
\begin{table}[h]
    \centering
    \begin{tabular}{c c c c c c c c}
        \toprule
        &  $\hat a$ & $a_1$ & $a_2$ & $b_1$ & $b_2$ & $b_3$ \\
        \midrule
        $c_i$ & $k>2$ & 1 & 1 & 0 & 0 &0        \\
        \bottomrule
    \end{tabular}
\end{table}

To describe the cost function for bundles with multiple items we set $A = \{\hat a, a_1, a_2\}, B = \{b_1, b_2, b_3\}$ and $B_{-i} = B \setminus b_i$. Now the cost function for agent $i$ is given by the following formula:
$$ c_i(X_i) = \begin{cases}
        {k^2}, & B_{-i} \subseteq X_i \text{ or } (b_i \in X_i \text{ and }X_i \cap A \ne \emptyset) \\
        \sum\limits_{x \in X_i} c_i(x), & \text{otherwise}
    \end{cases}$$
In words, agent 1 has a cost of $k^2$ for the bundle $\{b_2,b_3\}$, its supersets, and any bundle that contains $b_1$ paired with some chore from $A$. Otherwise, her cost function is effectively additive. We are now ready to prove our main theorem.

    Let $X$ be an allocation and consider an agent $i$, say agent 1 without loss of generality, that receives some item(s) from $B$. 

    \begin{itemize}
        \item Agent 1 receives 3 items from $B$\\
        Then we have that $\{b_2, b_3\} \subseteq X_1 \setminus  b_1 \implies c_1(X_1 \setminus e) = k^2 $ while $\min(c_1(X_2), c_1(X_3)) \le 2$ yielding a $k^2 / 2$ approximation ratio.
        \item Agent 1 receives 2 items from $B$
        \begin{itemize}
            \item Agent 1 receives $\{b_2, b_3\}$\\
            If she receives some extra item then again we have that $B_{-1} \subseteq X_1 \setminus  e$ and  $\min(c_1(X_2), c_1(X_3)) \le k$ giving a ratio of $k$. Thus, she should receive exactly $B_{-1}$. But then $c_2(X_1) =  c_3(X_1)= 0$ while at least one of them received multiple items, hence the EFX property is violated and the ratio is unbounded. 
            \item Agent 1 receives $\{b_1, b_3\}$ (symmetrically for $B_{-3}$)\\
            Again if she receives some extra item then $c_1(X_1 \setminus b_3) = k^2$ while $\min(c_1(X_2), c_1(X_3)) \le k$. Thus, she should receive exactly $B_{-2}$; but then, $c_3(X_1) = 0$. Therefore agent 3 should receive at most one item; in case this item is $b_2$, then agent 2 will receive all the non zero items while the other bundles have zero cost leading to unbounded ratio. But if agent $3$ does not receive $b_2$, then $X_2 = \{b_2, a, \hat a\}$ (where here, $a$ denotes either $a_1$ or $a_2$) or $\{b_2, a_1, a_2\}$ meaning that $c_2(X_2 \setminus a) = k^2$ while $c_2(X_3) \le k$. Once again the EFX property is violated and the ratio is $k$. 
        \end{itemize}
        \item Agent 1 receives 1 item from $B$
        \begin{itemize}
            \item Agent 1 receives $b_1$\\
            If she receives two more items from $A$ her cost is $k^2$  even after removing one item while $\min(c_1(X_2), c_1(X_3)) \le k$, while if she receives only $b_1$ then $c_2(X_1) = c_3(X_1) =0$ leading to a scenario analogous to when agent $1$ receives $B_{-1}$. It remains to check what happens when agent $1$ receives exactly one item from $A$.
            \begin{itemize}
                \item Agent $1$ receives $\{b_1, \hat a\}$\\
                Now $c_1(X_1\setminus e) = k$ and we are left with $a_1, a_2$ and $b_2, b_3$. No matter how we partition the remaining items among the rest of the agents, $\min(c_1(X_2), c_1(X_3)) \le 2$ giving us a $k/2$ approximation ratio. 
                \item Agent $1$ receives $\{b_1, a\}$, i.e. $b_1$ and one of $a_1$ and $a_2$ \\
                Now it is the other way around: $c_2(X_1) = c_3(X_1) = 1$ thus whoever gets $\hat a$ shall not receive more items; in any other case she would have been strongly envious, leading to a $k$ approximation ratio. Assume wlog that agent $2$ gets $\hat a$. Then agent $3$ gets $b_2, b_3, a_1$ (or similarly, $b_2, b_3, a_2)$ thus after removing $b_2$, her cost remains $k^2$, resulting in a $k^2$ ratio.
            \end{itemize}
            \item Agent 1 does not receive $b_1$\\
            Assume wlog that no agent received her matching item and again due to symmetry assume that $1 \gets b_2, 2 \gets b_3, 3 \gets b_1$. Furthermore, assume that 1 gets $\hat a$. Then $c_1(X_1 \setminus e) = c_1(\hat a) = k \geq \frac{k}{2} \cdot c_1(X_2)$ completing the proof.
        \end{itemize}
    \end{itemize}
    We conclude that the instance admits no EFX allocation and the approximation ratio is $k / 2 $ that grows unbounded as $k\to \infty$.
\end{proof}

We conclude the section with an immediate corollary of \Cref{th:non-ex}. 
\begin{corollary}[Maximin share guarantee implications]\label{cor:mms}
    The existence of a maximin share (MMS) fair allocation does not imply the existence of an $\alpha$-EFX allocation for any $\alpha \ge 1$.
\end{corollary}

\begin{proof}
    Notice that the MMS value for agent $i$ is $\mu_i^n(M) = k$, since proposing the allocation $X' = (\{k\}, \{a_1,a_2,b_l\}, \{b_i, b_j\})$ guarantees her MMS. Finally, the allocation $A = (\{k\}, \{a_1,a_2,b_1\}, \{b_2, b_3\})$ achieves the maximin share for each agent. However, the approximation ratio for EFX is unbounded, due to the fact that $c_2(A_2\setminus e) \geq 1, \forall e \in A_2$ while $c_2(A_3) = 0$. In fact, this implies something even stronger: the existence of a maximin share (MMS) fair allocation does not imply the existence of an $\alpha$-EF1 allocation for any $\alpha \geq 1$.
\end{proof}

\subsection{NP-Completeness}
In the sequel, we complement our negative results by studying the computational complexity of EFX allocations with indivisible chores under general monotone cost functions. We show that the problem is NP-complete. In fact, the crux of the non-existence argument lies in the inherent structure of the partition problem, which is well-known to be NP-complete. Some of the cases follow the analysis presented in \Cref{th:non-ex}; we present the full proof for completeness. 

\begin{theorem} \label{thm:npc}
    Deciding whether an EFX allocation exists for $3$ agents given a chores-only instance is NP-complete.
\end{theorem}

    \begin{proof}
    We will prove the statement by reducing from $3$-\textsc{Way Number Partitioning}, in which we are given a multiset $A = \{a_i\}_{i=1}^n$ of $n$ positive integers summing to $3S$ with $a_i < S$ and the question is whether a partition of $A$ into three sets of equal sum exists. Given an instance $I$ of 3-way partitioning, we construct an instance $I'$ of our problem that has 3 agents and $n+3$ items $M = \{a_1, \dots, a_n, b_1, b_2, b_3 \}$. We call $B = \{b_1, b_2, b_3 \}$ and $A = M\setminus B$. The costs of individual chores are the same for all agents $i\in \{1,2,3\}$: $c_i(a_j) = a_j$ for $j\in [n]$, and $c_i(b_j) = 0$ for $j\in \{1,2,3\}$. 
For a bundle $X_i$ assigned to agent $i$ where $|X_i|>1$, the agent's valuation is determined as follows: $$c_i(X_i) = \begin{cases}
        \infty, & B_{-i} \subseteq X_i \text{ or } (b_i \in X_i \text{ and } A\cap X_i \ne\emptyset),\\
        \sum\limits_{x \in X_i} c_i(x), & \text{otherwise}
    \end{cases}$$

For the forward direction, suppose that $I$ is a YES-instance of 3-way partitioning. Then, there exists a partition of $\{a_i\}_{i=1}^n$ into 3 sets, say $A_1, A_2, A_3$ such that $\sum_{x\in A_1}x=\sum_{x\in A_2}x=\sum_{x\in A_3}x$. We construct the following chore allocation for $I'$ and we will prove that it is EFX: for $i\in {1,2,3}$, agent $i$ is assigned the items from $A$ whose corresponding integers belong to $A_i$, along with an item $b_j$ where $j\neq i$. Then, for any two agents $i,j$ it holds that $c_i(X_i)=\sum_{x\in A_i}x=\sum_{x\in A_j}x=c_i(A_j) \leq c_i(X_j),$ and the EFX property follows. 

We now move to the reverse direction and we suppose that $I'$ is a YES-instance of our problem. Then there exists an EFX allocation, and we will demonstrate that $I$ is a YES-instance of 3-way number partitioning by initially establishing the following claim:

\begin{claim}\label{claim}
    The only plausible EFX allocations for $I'$ involve assigning exactly one item from $B$ to each agent $i$. This item should be $b_j$, where $j \neq i$.
\end{claim}

\begin{proof}

    To prove the claim we will eliminate every other feasible allocation by showing that it is not EFX. Note that for each agent, at most two bundles can have an infinite cost, implying that if there exists an agent $i$ with $c_i(X_i \setminus e) = \infty $, she will always be envious of another agent. We consider the following cases:
    \begin{itemize}
        \item Suppose that the EFX allocation involves an agent $i$ such that $|X_i\cap B|=3$. Then, $c_i(X_i)=\infty$ and $c_i(X_i\setminus b_i)=\infty$. At the same time, there exists an agent $j$ who is assigned only chores from $A$, say a set $A_j$. Thus $c_i(X_j)=\sum_{x\in A_j}c_i(x)<c_i(X_i\setminus b_i)$ and hence agent $i$ envies agent $j$. 
            \item Suppose that the EFX allocation involves an agent $i$ such that $|X_i \cap B|=2$ and $X_i \cap A \neq \emptyset$. There are two subcases to consider here based on whether $i$ receives $b_i$. \begin{itemize}
                \item In the first case in which she does receive $b_i$,  agent $i$ also receives a chore $b_j,$ for some $j\neq i$. Then, it holds that $c_i(X_i\setminus b_j)=\infty$ and there exists an agent $j$ who is assigned a (perhaps empty) set of chores from $A$, say a set $A_j$, and either no chores from $B$ or $b_k,$ for $k\neq i,j$. Hence, $c_i(X_j)= \sum_{x\in A_j}c_i(x)<c_i(X_i\setminus b_j)$ and agent $i$ envies agent $j$.
                \item For the second case in which agent $i$ receives $B_{-i}$ and at least one item from $A$, say $a$, it holds that $c_i(X_i\setminus a)=\infty$ and with a similar reasoning as in the previous case, we can show that there is an agent $j$ envied by agent $i$.
            \end{itemize} 
        
        \item Suppose that the EFX allocation involves an agent $i$ such that $|X_i \cap B|=2$ and $X_i \cap A = \emptyset$. Again, there are two subcases to consider based on whether $i$ receives $b_i$ or not. 
        \begin{itemize}
            \item In the first case, apart from $b_i$, agent $i$ also receives a chore $b_j,$ for some $j\neq i$. Then, for agent $j$ it holds that $c_j(X_i)=0,$ since $X_i$ contains $b_j$ and no items from $A$. If agent $j$ receives at least two items, she will envy agent $i$ after one of the items is removed. Otherwise, agent $j$ receives at most one item, and the third agent, referred to as agent $k$, either receives $b_k$ and at least two items from $A$, or agent $j$ receives only $b_k$. In the first case, if $a\in X_k$, it leads to $c_k(X_k\setminus a)=\infty>\sum_{x\in X_j}x=c_k(X_j)$, resulting in envy towards agent $j$. In the second case, $c_k(X_j)=0$ and, since once again agent $k$ receives at least two items from $A$, it follows that $c_k(X_j\setminus a)>c_k(X_j)$, demonstrating envy from agent $k$ to agent $j$.
            \item For the second case where agent $i$ receives $B_{-i}$ and no items from $A$, it holds that $c_j(X_i)=0$ for any agent $j$ other than $i$. Each of these agents will envy agent $i$ unless they receive at most one item, but taking into account that the total number of items that should be fully allocated among these agents is $n+1$, the considered allocation does not satisfy the EFX property. 
        \end{itemize} 
    \end{itemize}
    Therefore, the EFX allocation should allocate exactly one item from $B$ to every agent. Now, suppose that there exists an agent $i$ who has been allocated $b_i$. First, consider the case where $|X_i\cap A|>1$ and say that $a\in X_i\cap A$. Then $c_i(X_i\setminus a)=\infty$ and $c_i(X_j)=\sum_{x\in X_j}x<\infty,$ since $b_i\notin X_j$, leading to an envy from agent $i$ to agent $j$. On the other hand, let's focus on the case where agent $i$ receives at most 1 item from $A$. Say that agent $j \neq i$  receives a bundle $X_j$ and that $c_j(X_i)=a,$ where $a$ is the value that agents have for the item in $X_i \cap A$ (perhaps 0, if $X_i\cap A=\emptyset$). Since $a<S=\frac{\sum_{i\in [n]}a_i}{3}$, there should be an agent, say agent $j$, such that $a< \sum_{x\in X_j}x$. Agent $j$  also receives an item from $B$, say $b$. Thus, $c_j(X_j\setminus b) =\sum_{x\in X_j\setminus b}x =\sum_{x\in X_j}x > a= c_j(X_i)$, leading to an envy from agent $j$ to agent $i$.
\end{proof}

So the EFX allocation (if it exists) assigns a bundle $X_i$ to agent $i$ such that $X_i=b_j\cup A_i$, for some $j\neq i$ and some set $A_i\subseteq A$. Obviously it should hold that $\cup_{i\in \{1,2,3\}}A_i=A$ and that $A_i\cap A_j=\emptyset,$ for any $i,j\in \{1,2,3\}$. According to the EFX property,  $c_i(X_i\setminus e)\leq c_i(X_j)$ should hold for any $e\in X_i$ and any pair of agents $i,j$. Therefore, it should hold that $c_i(A_i)\leq c_i(X_j)= c_i(A_j)$, for any pair of agents $i,j$, which leads to $\sum_{x\in A_1}x= \sum_{x\in A_2}x=\sum_{x\in A_3}x$. Consequently, this implies that $I$ is also a YES-instance of Partition. 
    \end{proof}

\subsection{Extension to any number of agents}

It follows from the proof given in the previous section that the non-existence construction can be extended to an arbitrarily large number of items. Next, we show that with a slight tweak our proof idea implies strong NP-hardness for the case of an arbitrary number of agents.  

\begin{theorem}\label{th:nagents} Deciding whether an EFX allocation exists given a chores-only instance with an arbitrary number of agents is strongly NP-hard.
\end{theorem}
\begin{proof}
    We will prove the theorem via a reduction from the \textsc{3-Partition} problem, which is known to be strongly NP-hard \cite{GareyJohnsonNP}. An instance of \textsc{3-Partition} consists of a set $A = \{a_1, \dots, a_{3n}\}$ of $3n$ positive integers whose sum is $nT$, $a_i \in \big(\frac{T}{4}, \frac{T}{2}\big)$, and the goal is to partition the set into $n$ triplets such that each triplet sums exactly to $T$.
    The construction of the fair division instance is a direct extension of the one in \Cref{thm:npc}: we have $n$ agents and $4n$ chores. The set of chores consists of $3n$ main chores and $n$ special chores that cost zero to all of the agents, thus $M = \{a_1, \dots, a_{3n}, b_1, \dots, b_n\}$. The cost function of an agent $i$ is specified as follows:
    
    \[c_i(X_i) = \begin{cases}
        \infty, & \{b_j, b_k\} \subseteq X_i,\ j,k\ne i  \text{ or } (b_i \in X_i \text{ and } A\cap X_i \ne\emptyset),\\
        \sum\limits_{x \in X_i} c_i(x), & \text{otherwise}
    \end{cases}
\]

$(\Rightarrow)$ The forward direction is the same as in \Cref{thm:npc} so we focus on the reverse.

$(\Leftarrow)$ Suppose that there exists an EFX allocation $X$. We will start by showing the equivalent statement of \Cref{claim}. To that end, note that for each agent at most $\lceil \frac{n}{2} \rceil$ bundles can have an infinite cost. Therefore no agent might receive 3 or more items from $B$. Thus, we can write $N = N_0 \cup N_1 \cup N_2$ where agents in $N_k$ get $k$ items from $B$ and our aim is to show that $N_2 = \emptyset$. Assume, for the sake of contradiction, that $N_2 \neq \emptyset$. Firstly, note that an agent in $N_2$ cannot possess items from $A$ due to the same reasoning as in the proof of \Cref{claim}. Since $\{b_i, b_j\}$ has zero cost to agents $i$ and $j$, there are $2|N_2|$ agents that have zero cost for this bundle. At least half of those agents are not in $N_2$ and thus any one of them must receive only a single item. If agent $i$ belongs to $N_1$, i.e., receives only one item from $B$, then $c_j(X_i) = 0$, thus no agent can receive two or more items from $A$ contradicting the fact that $X$ is a complete allocation. Assume that agent $i$ receives one item from $A$ and let $j$ be an agent that receives at least 3 items from $A$ (note that such an agent is guaranteed to exists since there are $n$ agents and $3n$ items). Then \begin{equation}\label{eq:star}
    c_j(X_j \setminus a) > 2 \cdot \frac{T}{4} = \frac{T}{2} > c_j(X_i)\tag{$*$}
\end{equation} 
violating the EFX property. Overall, $N_2$ is empty, and hence $N_0$ is empty as well.

Having established that each agent receives exactly one item from $B$, we assume that agent $i$ receives $b_i$. Then either $i$ receives multiple items from $A$, hence $c_i(X_i \setminus a) = \infty $ or $X_i = a \cup b_i$. The former case trivially contradicts EFX since agent $i$ has non-infinite cost for the other bundles. In the latter case, there is again an agent $j$ that receives 3 items from $A$ and \Cref{eq:star} holds, so EFX is violated.

Since no agent receives her corresponding item from $B$, we can write $X_i = A_i \cup b_j$ for $j \ne i$ and let $b_i \in X_k$. Then
\[c_i(X_i \setminus b_j) = c_i(A_i) \le c_i(X_j) = c_i (A_j), \forall j \ne k  \]
Combining the inequalities for each $i$ we get that $c(A_i) = c(A_j)$ for every pair $i,j$, which completes the proof.
\end{proof}

%% file: fewItems.tex
\section{EFX with a Few Chores} \label{sec:few-items}

Following the negative result of \Cref{subsec:31} (three agents, six chores) it is natural to ask whether an EFX allocation always exists with less items. We answer this question affirmatively for the case of $m \leq n+2$ chores and any number of agents. We note that this is the best achievable guarantee for the special case of three agents. The ideas to be presented revolve around the following simple fact.

\begin{observation}\label{obs}
    An agent who receives her two smallest chores does not strongly envy any nonempty bundle. 
\end{observation}

\begin{proof} 
Recall that $\sigma_i(j)$ denotes the $j$-th larger chore according to agent $i$. Now, for any $k < m-1$ we have $c_i(X_i\setminus e) \le c_i(\sigma_1(m-1)) \le c_i(\sigma_1(k)) \le c_i(X_j)$    
\end{proof}

We are now ready to state the following theorem.

\begin{theorem}\label{th:n+2}
    There exists an EFX allocation when $m \le n+2$ for any number of agents with general monotone valuations. 
\end{theorem}
\begin{proof}
    If $m\le n$ then we have an EFX allocation by allocating at most one chore to each agent. If $m=n+1$ then let $Z_i$ be the two smallest chores of agent  $i$. Allocating $Z_1$ to agent $1$ and one chore to each of the remaining agents arbitrarily is EFX due to \Cref{obs}. We focus on $m=n+2$. At a high level, the idea is the following: if there is a pair of agents $i, j \in N$ such that $Z_i \cap Z_j = \emptyset$, we allocate $Z_i$ to $i, Z_j$ to $j$ and again one chore to the other agents arbitrarily. Asking for a pair of disjoint $Z$'s is strict and can be relaxed. For instance, if we allocate $Z_1$ to agent $1$ and the two smallest chores chores from $M\setminus Z_1$ to agent $2$, then agent $2$ is EFX towards agents 3 to $n$ due to \Cref{obs}. Therefore, strong envy might only arise from agent $2$ to agent $1$. In order to avoid this scenario, it suffices that $Z_1$ contains one large chore for agent $2$. Concretely, we have the following two cases:
    \begin{itemize}
        \item For some agent $j: \sigma_j(k) \in Z_i$ for some $k\le n-1$. We construct an EFX allocation as follows: agent $i$ gets $Z_i$, agent $j$ gets her two smallest chores from $M\setminus Z_i$ and each remaining agent gets exactly one chore arbitrarily. It suffices to check that $j$ does not strongly envy $i$. Indeed, there are three chores with index larger than $k$, i.e. smaller than $\sigma_j(k)$, and agent $i$ received at most one of them. Thus, both chores of $j$ are smaller than $\sigma_j(k) \in X_i = Z_i$ and the case is completed.
        \item No such agent exists. Then $Z_i$ contains two out of the three smallest chores for any other agent. Consider the set $M\setminus Z_i$.
        \begin{itemize}
            \item There exists an agent $j$ such that $\sigma_i(n, M\setminus Z_i) \ne \sigma_j(n, M\setminus Z_i)$. In other words agents $i$ and $j$ disagree on the smallest chore in $M\setminus Z_i$. Then we allocate $\sigma_i(n, M\setminus Z_i)$ to $i$ and $\sigma_j(n, M\setminus Z_i$) to $j$, one item from $Z_i$ to each one and a single item to the other agents arbitrarily. Note that both agents $i$ and $j$ have received two out of their three smallest chores while the rest of the agents have received a strictly larger one. Finally, agent $i$ does not stongly envy $j$ because $\sigma_j(n, M\setminus Z_i)$ has a smaller index than $n$ for agent $i$ and vice versa.  
            \item No such agent exists. That means that all agents agree on the least costly chore in $M\setminus Z_i$. Since $Z_i$ contains two out of the three smallest chores for all agents, we conclude that all agents agree upon the set of the three smallest chores, namely $M^-$. If agent $1$ can receive $M^-$ without breaking the EFX property the case is completed. Assuming the contrary, we have that there exists a subset $Y \subset M^-$ with $\lvert Y \rvert = 2$, such that $c_1(Y) > c_1(\sigma_1(n-1))$. Therefore the allocation in which agent $1$ gets $\sigma_1(n-1)$ and $M^- \setminus Y$, some agent gets $Y$, and all other agents get one item arbitrarily is EFX. \qedhere
        \end{itemize}
    \end{itemize}
\end{proof}

%% file: addApx.tex
\section{Approximations for Additive Cost Functions}\label{sec:addapx}
We now switch to the case of additive cost functions, aiming to obtain better approximation guarantees. To that end, we revisit the Envy Cycle Elimination algorithm. So far, this procedure is less explored in the chore setting due to its lack of monotonicity. We try to bypass this obstacle via a general approximation framework due to \cite{Farhadi_Hajiaghayi_Latifian_Seddighin_Yami_2021,markakis2023improved}. These works study the goods-only setting, and in fact, all of the results mentioned therein do not carry over to chores. \cite{markakis2023improved} were able to beat the state of the art for approximate EFX for goods under ordinal assumptions. Whether improved EFX approximations can be achieved in general is an intriguing open question. 

\subsection{Approximation Framework} \label{sec:apx-framework}

{
\setlength{\belowdisplayskip}{0pt}
\begin{algorithm}[tb]
\caption{Chore approximation framework}\label{alg:frame}
    \textbf{Input}: $N, M, C$\\
    \textbf{Output}: An allocation $X$
\begin{algorithmic}[1]
\STATE{For $\alpha, \beta >0$, compute a partial $\alpha$-EFX allocation $Y = (Y_1,\dots, Y_n)$, with the property that 
$$ c_i(e) \le \beta \cdot c_i(Y_j) 
\text{ for all $i, j\ne i \in N$ and all } e \in M\setminus Y$$}
\STATE{Run  \Cref{alg:ttece} until there are no unallocated items} 
\end{algorithmic}
\end{algorithm}}
\begin{theorem}\label{theorem:frame}
The allocation produced by \Cref{alg:frame} is $\max{(\alpha, \beta+1)}$-EFX. 
\end{theorem}

Before proceeding with the proof, we make a small observation and introduce some additional notation.

\begin{observation}\label{obs:rp}
Consider a partial allocation $Y$ and define the ratio matrix $R = \{r_{ij}\}$ as follows $$r_{ij} = \frac{c_i(Y_j)}{\max\limits_{e \in M\setminus Y}c_i(e)}$$
In the case of goods one produces approximations by considering only the diagonal of the matrix $R$ while for chores we need to consider anything but the diagonal. 
\end{observation}
In the sequel, we refer to the property in line 1 of  \Cref{alg:frame} as the \say{ratio property} and stick with the $r_{ij}$ notation.

\begin{proof}[Proof of \Cref{theorem:frame}]
Fix some agent $i$. Initially, i.e. in $Y$, any envy from agent $i$ is at most $\alpha$-EFX. When a chore is not allocated to agent $i$, we get $r'_{ij} \ge r_{ij}$ since the nominator increases or the denominator decreases or there are no changes. Thus $r'_{ij} \ge \beta$ is maintained. Moreover, agent $i$ is indifferent to cycle eliminations she does not participate in; such procedures simply permute the line $R_i$ without affecting $r_{ii}$.   
Therefore, it remains to check what happens when she does participate in a cycle elimination or when she gets some item $e$. Let $Z$ be the allocation in the first scenario and $Z_j$ the bundle she is about to receive. Then we have that $c_i(Z_i) \ge c_i(Z_j) \implies r_{ii} \ge r_{ij} \ge \beta$. Therefore, line $R_i$ is element-wise greater than $\beta$ and thus, no bundle reallocations can disrupt her ratio property. Finally, when $e$ gets allocated to $i$ we have:
\begin{alignat*}{2}
c_i(Z_i) &\le &&c_i(Z_j)\\
c_i(e) &\le \beta \cdot &&c_i(Z_j) 
\end{alignat*}
The first inequality holds since $i$ is now a sink in the envy graph, while the second follows from the ratio property, as described above. Adding the inequalities yields the $\beta+1$ term claimed, completing the proof.
\end{proof}

\subsection{Conditional Approximations}
Similarly to the goods-only setting, one can obtain guarantees by examining only the few most important chores. Here, the most important ones are those with high costs for some agents. 
\begin{definition}\label{def:Ls}
    Let $L_i^k = \{\sigma_i(1), \dots, \sigma_i(k)\}$ denote the set of the $k$ most burdensome chores for agent $i$.
\end{definition}
In accordance with the goods setting, we refer to chores in $L_i$ as \say{top}.

\begin{lemma}[Top $n$ agreement]\label{lem:nagr}
If $L_i^n = L_j^n$, i.e. all agents agree in the set of the top $n$ chores, there exists a $2$-EFX allocation. Moreover, it can be computed in polynomial time.    
\end{lemma}

\begin{proof}
Note that allocating one chore from the set to each agent arbitrarily, produces an EFX allocation that satisfies \Cref{theorem:frame} with $\alpha=1$ and $\beta=1$, thus yielding a 2-EFX complete allocation. Moreover, since the partial allocation can be constructed in linear time the whole procedure is efficient. 
\end{proof}

\begin{lemma}[Top $n-1$ agreement]\label{lem:n1agr}
If $L_i^{n-1} = L_j^{n-1}$, i.e. all agents agree in the set of the top $n-1$ chores, there exists a $\max(2, n-2)$-EFX allocation. Moreover, it can be computed in polynomial time.
\end{lemma}

\begin{proof}
    Let $L^{n-1} = \{l_i\}_{1\le i\le n-1}$ be the top set and construct the allocation $X$ as follows:
    $$X_i = \begin{cases} 
    l_i, &1\le i \le n-1\\
    \bigcup_{j\ne i} \sigma_j(n), &i=n
    \end{cases}$$
    Agents 1 to $n-1$ are EFX towards the rest since they have a single item. If agent $n$ also has a single item, we have $\alpha=1, \beta=1$ thus a 2-EFX allocation. Assuming that agent $n$ has multiple items we have $$ c_n(X_n\setminus e) \le (n-2)\cdot c_n(\sigma_n(n-1)) \le c_n(X_j) $$ where the first inequality is due to the fact that $n$ got at most $n-1$ items. Thus $\alpha=n-2$. As for the ratio property, for all agents but $n$, $r_{ij} \ge 1$ since every bundle contains one chore from $L_i^n$ and for agent $n: r_{ij\ne i} \ge 1$ since every other bundle contains one chore from $L_n^{n-1}$. Again $\beta = 1$ and the proof is complete.  
\end{proof}

Interestingly, the same approximation ratio can be obtained when the agents exhibit diametrically opposed preferences, i.e. disagree on all top $n-1$ chores. Note, however, that this implies that there exist at least $n(n-1)$ items.
\begin{algorithm}
\caption{Top $n-1$ disagreement}\label{alg:n1dis}
\begin{algorithmic}[1]
\FOR{each agent $i$}
    \FOR{each agent $j\ne i$}
        \STATE{$e^*  = \arg\max_{e \in M} c_i(e)$}
        \STATE{$X_j = X_j \cup e^*$}
        \STATE{$M = M \setminus e^*$}
    \ENDFOR
\ENDFOR
\end{algorithmic}
\end{algorithm}

\begin{lemma}\label{lem:n1dis}
If $L_i^{n-1} \cap L_j^{n-1} = \emptyset$, i.e. all the agents disagree in the set of the top $n-1$ chores, there exists a $\max(2,n-2)$-EFX allocation. Moreover, it can be computed in polynomial time.    
\end{lemma}

\begin{proof}
    We will show that the partial allocation produced by  \Cref{alg:n1dis} satisfies the conditions of  \Cref{theorem:frame} with $\alpha=n-2$ and $\beta=1$. We have that $$ c_i(X_i\setminus e) \le (n-2)\cdot c_i(\sigma_i(n-1)) \le c_i(X_j)$$ since agent $i$ cannot receive any item from $L_i^{n-1}$ while every other agent receives exactly one item from it. The latter fact also implies that $r_{ij} \ge 1$.
\end{proof}

\paragraph{Connections between goods and bads.}
So far, results from goods naturally extend to the chores setting. However, one noteworthy difference is that in the goods setting, both the agreement and disagreement assumptions yield the desired $2/3$ approximation. Conversely, in our case, there is a gap between $2$ and $n-2$ as the agents transition from agreement to disagreement. Additionally, a result that, perhaps surprisingly, does not translate to our setting is the truncated common ranking assumption (Corollary 1 in \cite{markakis2023improved}). While an EFX allocation is known to exist when the agents agree upon the ranking of the items, slightly relaxing the assumption to an agreement in all but few items cannot guarantee something better than 2-EFX, as shown below. 

\begin{observation}\label{lem:cntr}
Even if the agents agree upon the ranking of a large set $S$ of top chores, and an exact EFX allocation of $S$ is given, we cannot obtain better guarantees via \Cref{alg:frame}. \end{observation}

\begin{proof}

    Consider the instance with $k$ common top chores, with $2n < k < m$, as shown below:

    \begin{table}[h]
    \centering
    \begin{tabular}{c|c|c|c|c|c|c}
        & $e_1$ & $\dots$ & $e_{n-1}$ & $e_n$ & $\dots$ & $e_k$\\\hline
        1 & 1 & $\dots$ & 1 & 1 & $\dots$ & 1 \\
        \vdots & \vdots & \vdots & \vdots & \vdots & \vdots &  \\
        $n-1$ & 1 & $\dots$ & 1 & 1 & $\dots$ & 1 \\
        $n$ & $M$ & $\dots$ & $M$& $\frac{1}{k+1-n}$ & $\dots$ & $\frac{1}{k+1-n}$ \\
    \end{tabular}
\end{table}
It is easy to verify that the allocation $\left(e_1, \dots, e_{n-1}, \bigcup_{i=n}^k e_i\right)$ is envy free and thus EFX. Additionally, note that there is nothing remarkable about the allocation, in the sense that it can produced by an execution of the TTECE. However, for agents $1$ to $n-1$ we cannot rule out the possibility that $c_i(e) = 1-\epsilon$ for some unallocated chore $e$, thus we cannot guarantee $\beta > 1 + \epsilon$.  
\end{proof}

\subsection{Approximation for Three Agents}\label{sec:3agents2}

Next, we treat the case of three agents with additive disutilities. We show how to apply the techniques developed in the previous section to obtain unconditional results for the case of three agents, improving the approximation factor from $2+\sqrt{6}$ to 2.

\begin{theorem}\label{th:3agents}
 A 2-EFX allocation for three agents exists and can be computed in polynomial time.    
\end{theorem}

\begin{proof}
    Due to \Cref{lem:n1dis}  there exists a pair of agents that agree upon at least one top-2 item; otherwise the theorem follows immediately. Without loss of generality assume that agents 1 and 2 agree. If they agree only on the second chore, i.e. $\sigma_1(2) = \sigma_2(2)$ we allocate it to agent $3$ and agents $1$ and $2$ receive each other's top chore. Otherwise, agent $1$'s top chore lies in agent $2$'s top-2 set (or vice versa). We allocate it to agent $3$. Then we allocate $\sigma_2(1, M \setminus X_3)$ to agent 1 and $\sigma_1(1, M \setminus (X_2 \cup X_3))$ to agent 2. Now, the allocation is trivially EFX since any agent has exactly one item. Crucially, due to the allocating order, agents 1 and 2 do not envy agent 3. As for the ratio matrix we have: 
    $$ R = \begin{bmatrix}
        \bullet & \ge 1 & \ge1\\
        \ge 1 & \bullet & \ge 1\\
        r_{31} & r_{32} & \bullet\\
    \end{bmatrix}$$
    All that is left is to ensure that $r_{31}$ and $r_{32}$ can be made larger than 1. To that end, note that running TTECE using agent 3's cost function and picking the sinks lexicographically will allocate to agent 3 only after allocating one more item to the other agents, therefore at that point $r_{3i} \ge 1$ fulfilling the requirements of \Cref{theorem:frame}. In case that does not happen, the resulted allocation is complete; agent 3 has a single chore, thus she cannot strongly envy, while agents 1 and 2 satisfy the requirements of \Cref{theorem:frame} with $a=1$ and $b=1$.
\end{proof}

\paragraph{Remark} To the best of our knowledge, this is the first instance where the Envy Cycle Elimination algorithm (either for goods or chores) is used in such a manner, i.e. picking both the source/sink and the item in a specific way. 

%% file: conclusion.tex
\section{Conclusion and Future Work}
We explore EFX allocations in the context of bads. We demonstrate a series of strong negative results regarding the existence, approximation, and computation of allocations satisfying EFX. Moreover, we show that EFX always exists under a setting with a small number of items, and provide a separation result with the goods-only setting. Lastly, we show improved approximation ratios for a number of cases. Our work leaves two main open questions. First, determining whether similar constructions can be found for the case of goods. Second, whether an exact EFX allocation always exists for three agents with additive disutilities.

%% file: acknowledgments.tex
\section*{Acknowledgements}

We would like to thank Georgios Amanatidis, George Christodoulou, Evangelos Markakis, and Alkmini Sgouritsa for helpful discussions and their constructive comments on a preprint of this work. We are also grateful to
Georgios Papasotiropoulos and Nicos Protopapas for their assistance and their insightful feedback. Lastly, we thank the anonymous reviewers for their helpful comments.

This work has been partially supported by project MIS 5154714 of the National Recovery and Resilience Plan Greece 2.0 funded
by the European Union under the NextGenerationEU Program.